\documentclass[12pt]{article}
\newread\raux
\openin\raux=srcltx.sty\relax
\ifeof\raux
\closein\raux
\else
\closein\raux
\usepackage[active]{srcltx}
\fi
\usepackage{cite}
\usepackage{amssymb}
\usepackage{amsthm}
\usepackage{amsmath}
\usepackage{epsfig}
\usepackage{graphicx}

\DeclareMathOperator*{\argmax}{argmax}

\newtheorem{proposition}{Proposition}[section]

\newtheorem{lemma}{Lemma}[section]
\theoremstyle{definition}
\newtheorem{defini}{Definition}[section]
\newenvironment{proofsketch}{\trivlist\item[]\emph{Proof Sketch}:}%

\newcounter{linenumber}
\newcounter{indentation}
\newcommand{\algosize}{\scriptsize}
\newcommand{\capnir}{}

\newcommand{\R}{\mathbb{R}}
\newenvironment{niralgo}[4]{
\setcounter{linenumber}{0}
\begin{figure}[#4]
  \centering
  \renewcommand{\capnir}{\caption{#1}\label{#2}}
  \algosize
  \begin{tabular}{|l|}
    \hline\\
    \begin{minipage}{#3} 
      \begin{tabbing}
}
{
      \end{tabbing}
    \end{minipage}\\
    \hline
  \end{tabular}
  \normalsize
  \capnir
\end{figure}
}

\makeatletter
\gdef\linelabel#1%
{%
  \immediate\write\@auxout{
    \string\newlabel{#1}{%
{\thelinenumber}%
{\thepage}%
}%
}%
}
\makeatother

\newcommand{\algorule}{\hspace{1mm}\vline}
\newcommand{\num}{
  \stepcounter{linenumber}%
{\tt \thelinenumber} \>}

\newcommand{\numm}{\num \algorule \>}
\newcommand{\nummm}{\num \algorule \> \algorule \>}

\newcommand{\remove}[1]{}

\newcommand{\kkk}[1]{
\begin{array}{cc}
\left[\begin{array}{cc}
\; & \; \\
\; & \; 
\end{array} \right]_{k\times k} & \rightarrow \overline{s_{#1}} \end{array}
}

\begin{document}
\title{Efficient and Universal Corruption Resilient Fountain Codes}
\author{
Asaf~Cohen, Shlomi~Dolev and Nir~Tzachar
\thanks{Asaf Cohen is with the Department of Communication Systems Engineering, Ben-Gurion University of the Negev, Beer-Sheva, 84105, Israel, {\tt coasaf@bgu.ac.il}}
\thanks{Shlomi Dolev and Nir Tzachar are with the Department of Computer Science, Ben-Gurion University of the Negev, Beer-Sheva, 84105, Israel, {\tt {tzachar,dolev}@cs.bgu.ac.il}.} 
\thanks{A version of this manuscript appeared as a brief announcement in DISC 2008.}
}
\date{}
\maketitle
\begin{abstract} 
In this paper, we present a new family of fountain codes which overcome adversarial errors. That is, we consider the possibility that some
portion of the arriving packets of a rateless erasure code are corrupted in an undetectable fashion. In practice, the corrupted packets may be
attributed to a portion of the communication paths which are controlled by an adversary or to a portion of the sources that are malicious.

The presented codes resemble and extend LT and Raptor codes. Yet, their benefits over existing coding schemes are manifold. First, to overcome the corrupted packets, our codes use information theoretic techniques, rather than cryptographic primitives. Thus, no secret channel between the senders and the receivers is required. Second, the encoders in the suggested scheme are oblivious to the strength of the adversary, yet perform as if its strength was known in advance. Third, the sparse structure of the codes facilitates efficient decoding. Finally, the codes easily fit a decentralized scenario with several sources, when no communication between the sources is allowed.   

We present both exhaustive as well as efficient decoding rules. Beyond the obvious use as a rateless codes, our codes have important applications in distributed computing.
\end{abstract}


\section{Introduction}\label{sec. intro}
Modern erasure codes allow efficient encoding and decoding schemes for use in a wide range of applications. Recent schemes include {\em digital fountain} codes \cite{BLMR}, Raptor codes \cite{RAPTOR} and LT codes \cite{LT}. These codes are rateless erasure codes, in the sense that the encoder produces a practically infinite stream of packets, such that the original data can be efficiently decoded from any sufficiently large subset. These codes allow linear time encoding and decoding, with high
probability. 

Nevertheless, to the best of our knowledge, such codes are only resilient
to {\em lossy channels}, where packets may be lost but not arbitrarily
corrupted. As a consequence, such solutions cannot cope with {\em
Byzantine} channels, where an adversary may arbitrarily change packets; for
instance, when communication is done over parallel paths and some paths are
under the control of an adversary.

Such attacks cannot be attended to using standard error correcting codes. Consider, for example, a Raptor code, where encoded packets are sent in order
to transfer a message. Each packet sent is further encoded by using an error
correcting code. Consider an adversary which corrupts a single packet and the
packet's error correction padding, such that the error correcting code does not
identify the packet as a corrupted packet. This corrupted packet may well
prevent the receiver from correctly decoding the entire original message.

In this work, we design and analyze novel erasure codes which are capable of withstanding malicious packet corruption. Our coding scheme resemble LT codes, yet are information theoretic secure (in a sense to be rigorously defined).
\subsection{Applications}
Corruption resilient fountain codes have numerous applications. We present here but a few.

\noindent {\bf Erasure coding}. Consider a Peer-to-Peer system, where a user
would like to receive some content. To alleviate the load on any single Peer,
the content may be mirrored at several peers. On the other hand, to maximize
bandwidth usage, the user should be able to receive parts of the content from
several mirrors in parallel. Erasure codes, and in particular digital fountain
codes, give rise to a simple solution to this problem; each mirror locally,
and independently, generates packets and sends them to the user. 

Alas, the system described above is very sensitive to Byzantine peers; when even
one of the mirrors intentionally corrupts packets, the receiver will never be
able to reconstruct the requested content. A more robust solution is to use
the corruption resilient fountain codes derived herein, such
that a constant fraction of Byzantine mirrors can be tolerated.

\noindent {\bf Shared value}. Consider a group of sensors
which receive and record global inputs. The inputs may be from a control and command
entity, such as a satellite, or a natural event that the sensors sense.
The sensors wish to store such inputs (or some global history) for later retrieval. 
Assume the initially shared value $x$ includes $k$ bits.
We wish to reduce the storage requirements at each sensor, such that each sensor
will only need to store a fraction of the $k$ bits. We require
that {\it no communication} takes place during the initial stage, so that each
sensor generates its own encoded (short) share of $x$ independently of other sensors.
The sensors may communicate later to reconstruct $x$ from the stored shares. 
Moreover, the solution should be robust against a
constant fraction of Byzantine sensors, where a Byzantine sensor may introduce
arbitrary information into the system.

Corruption resilient fountain codes can be used to solve the {\it shared value} problem. We present a randomized scheme in which shared data is efficiently recorded with {\it no communication} among the sensors. Note that, in some cases, it is also possible to update the encoded data without decoding (e.g. \cite{DLY07}).
\def\USINGECC{
A different approach to the one taken above can be achieved by using error
correcting codes (for example, Reed-Solomon code), when nodes have, or can
randomly choose, distinct identifiers. Namely, encoding a shared value $x$ at
each node by using a $[k, k-2f]$ code in the following way: let $k\ge3f+1$. Each
node builds a polynomial $g$ from $x$, of degree at most $k-2f$ (by padding $x$
with zeros as necessary), in a deterministic way, such that each node builds
the same $g$. Finally, each node $p$ evaluates $g$ at a point which corresponds
to $p$'s identifier, e.g., $g(id_p)$, stores the result and discards $x$.

When a node, $p$, wishes to reconstruct $x$, $p$ needs to collect $k$ distinct
values from other sensors (and $p$), which, together with the value stored
locally at $p$, can be used to reconstruct $x$, overcoming at most $f$ faulty
values (possibly caused by Byzantine nodes). Such $k$ distinct values may be
collected when, for example, $p$ have $k$ neighbors.
}
\subsection{Related Work}\label{subsec. related}
The current literature includes several different strategies for coping with Byzantine adversaries, both in erasure coding and network coding scenarios.
A common approach to overcome Byzantine adversaries when implementing erasure
codes is to check each received packet against a pre-computed hash value, to
verify the integrity of the packet. When using fixed rate codes, the sender can
pre-compute the hash value of each possible packet and publish this hash
collection in a secure location. The receiver first retrieves this pre-computed
hash collection and verifies each packet against the packet's hash as the
packet arrives. The hash is a one way function and, therefore, when the
adversary is computationally limited, the adversary cannot introduce another
packet with the same hash. However, when using rateless codes, such techniques are not feasible; as there is practically an infinite number of different packets, there is no efficient way to pre-compute the hash value of each possible packet and send these hashes to the receiver. Furthermore, inherent to to hashing technique is secure publication of the hashes. The sender must devise a way to securely transfer hashes to the receiver, say, by an expensive flooding technique. In this work, we completely avoid the need for such a secret channel.

In \cite{KFM04}, a slightly different technique for packet
verification in rateless codes is used. Therein, a Merkle-tree~\cite{merkle} based signature structure is suggested. However, the solution
proposed is, still, only valid against computationally bounded adversaries and
relies on the existence of homomorphic, collision-resistant hash functions.
Furthermore, as the size of a Merkle-tree is linear in the size of the original
message, the authors propose a process of repeated hashing to reduce the size of
the tree. Such recursive application of a hashing function is more likely to be
susceptible to attack.
 
An efficient scheme for signature generation of rateless codes appears in
\cite{ZKMH}, where the authors use the computational hardness of
the discreet log to provide a {\sc pki} which enables the sender
to efficiently sign each packet transmitted. The scheme is based on looking at
the data being sent as spanning a specific vector space, and looking at packets
as valid as long as they belong to the same vector space. The verification part
uses standard cryptographic devices to facilitate the check.

In this paper, however, we provide an \emph{information theoretically} secure rateless erasure code, by introducing encoding and decoding techniques which have a provable low probability of not recovering from an attack, assuming sufficiently many packets are collected.

Although not immediately apparent, network codes are closely related to rateless
erasure codes. Erasure codes may benefit from techniques to cope with Byzantine
adversaries developed for network coding and vice versa. Several network coding
related papers discuss the merits of using hash functions to overcome Byzantine
adversaries in network coding protocols, all of which require out of band
communication or preprocessing. Other protocols employ some kind of shared secret
between the sender and receiver to cope with computationally bounded
adversaries.

A different solution appears in~\cite{JLKHKM} where the only
assumption needed to overcome an all-powerful adversary is a shared value
between the sender and the receiver, which may also be known to the adversary. This shared value is, in fact, a parity check matrix, with which a sender inserts redundant bits to the data, enabling the receiver to identify the single correct message from a list generated by a \emph{list decoder}. Using
sufficient redundant information, the receiver can overcome a Byzantine adversary, as long as the adversary cannot control more than half of the network's capacity (which is the minimal cut between the sender and receiver). However, this solution requires, yet again, out of band communication or preprocessing, as the sender and receiver must share a value before starting the communication. More importantly, the redundancy matrix, which controls the amount of the redundant information inserted, needs to be \emph{known at the sender in advance}. In other words, a sender needs to know how strong the adversary is (how many packets it controls) in order to know how much redundant information to insert. In the solutions we present herein, the sender is completely oblivious to the strength of the attacker (or its actual existence). Finally, note that the solution presented in~\cite{JLKHKM} requires batching packets into groups of predetermined size and is not rateless in nature, especially not when a few senders are involved.

In \cite{KK}, Koetter and Kschischang present a different 
approach, based on high dimensional vector spaces; a message of $m\cdot k$ bits
is encoded into a vector space, $V$, of dimension $l\leq m$, which is a subspace
of an ambient vector space $W$ of dimension $l+m$. $l$ is a parameter of the
encoding scheme, $m$ is the number of bits in a message block and $k$ is the
number of blocks in the message. Each packet the sender creates is a
randomly chosen vector (of $l+m$ bits) in $V$. The receiver, upon collecting
enough vectors -- $l$ linearly independent vectors -- can proceed to
reconstruct the original message from the received vector space $U$. The authors
present a minimal distance decoder, which can recover $V$ from $U$ provided
that, when writing $U$ as $U=V\cap U + E$ where $E$ is the error space,
$t=dim(E)$ and $\rho=dim(V\cap U)$, it holds that $\rho+t<l-k+1$.

The codes presented in \cite{KK} have theoretical merits. Nevertheless, they
suffer from several severe implementation problems; to boost the error
resiliency of the code, one should (a) increase $l$ or (b) decrease $k$. The
implications are the need for sending more redundant information in each packet
(increasing $l$) and having larger packet sizes (decreasing $k$). Moreover, to
recover from a Byzantine attack on at most a third of the packets sent, the
codes presented require that $m\in\Omega(\sqrt{n})$ (that is, each block must be
of length $\Omega(\sqrt{n})$ bits), where $n$ is the size of the message. In
contrast, our codes are not limited by block sizes and can cope with one third
of Byzantine corrupted packets regardless of the block size.
\subsection{Main Contributions}\label{subsec. main cont}
In this work, we design and analyze rateless codes with the following merits. First, they are resilient to Byzantine attacks. Under some mild constraints, they asymptotically achieve the optimum rate of $C-2f$, where $C$ is the channel (or network) capacity, and $f$ is the number of corrupted packets. Second, the codes use sparse encoding vectors, enabling a decoding complexity of $O(k^2\log k \log\log k)$ instead of the usual $O(k^3)$ for non-sparse codes. Third, the encoding scheme carried out at the sources does not depend on the strength of the adversary (the number of packets it can corrupt), and hence is universal in this sense. Forth, the codes do not require any secrete channel or shared data between the sources and receivers. Moreover, no communication between the sources is required in case a few sources cooperate to send a common global value.  

The rest of the paper is organized as follows. The system settings and attack strategies on existing codes appear in Section \ref{s:ss}. Our new coding schemes appear in Section \ref{s:local}. The paper is concluded in Section \ref{s:cr}. 
\section{System Settings and Attack Strategies}
\label{s:ss}
A rateless erasure code is defined by a pair of algorithms, $\mathcal{(E,D)}$,
such that, given a set of input symbols, $\mathcal{E}$ can produce a practically
infinite sequence of different output {\em packets}. Moreover, from any large
enough set of packets, $\mathcal{D}$ can be used to recover all of the
original symbols.
A rateless erasure code is usually used between a {\em sender} and a {\em
receiver}, where the sender wishes to send a specific message to the receiver.
The sender starts by dividing the message into symbols, and then uses
$\mathcal{E}$ to generate packets, which are then sent to the receiver over a
lossy channel. The receiver, after collecting enough packets, uses $\mathcal{D}$
to recover the original symbols and, from them, the message.

Next, we define the adversarial model we use. We assume that the computational
power of the adversary is \emph{unlimited}, and the adversary may \emph{sniff all traffic} in the network. Furthermore, the adversary may forge or alter packets such that the receiver cannot differentiate them from legitimate packets. The only restriction we place on the adversary is the number of packets the adversary may corrupt. The restriction is defined by looking at all packets arriving at the receiver.
We then say that an adversary is $c$-bounded, with a parameter $c\leq
\frac{1}{3}$, if, for each $i\ge4$ and for each set of packets collected by the
receiver of size $i$, no more than $c\cdot i$ packets are corrupted. This
property captures the ratio between the number of collected packets
and the number of errors allowed.

Given the above settings, we discuss possible ways an adversary may
influence the Belief Propagation decoding algorithm used by \cite{LT,RAPTOR}.
Belief Propagation decoding suits the following succinct encoding
algorithm: to generate a packet, choose a random subset of input symbols and {\sc
xor} them. The exact distribution from which input symbols are sampled forms the
critical part of the encoding algorithm, and defines the number of packets
needed for correctly decoding the input symbols.

Belief Propagation decoding then works as follows: given a set of
packets, define a bipartite graph, $G=(A,B,E)$, where the bottom layer, $A$,
contains the packets and the upper layer, $B$, the input symbols. An edge exists
between a packet $a\in A$ and a symbol $b\in B$, if $b$ was used in generating
$a$. The Belief Propagation decoder is described in Figure~\ref{algo:BFD}, where
the successful completion of the decoding process depends on the neighbor
distribution.
\def\ALGA{
\begin{niralgo}{Belief Propagation decoder}{algo:BFD}{8cm}{h}
xxx \= xxx \= xxx \= xxx \= xxx \= \kill
definitions:\\
$N(a) = \{ b \in A\cup B|(a,b)\in E\}$\\
\\
\num {\bf While}$(\exists b\in B : |N(b)| > 0)$\\
\numm Let $p$ be a packet, such that $N(p) = \{s\}, s\in B$.\\
\numm Copy $p$ to its only neighbor, $s$, which is then successfully
decoded.\\
\numm {\bf For each} $a\in N(s)$ {\bf do}\\
\nummm $a \gets a \oplus s$ \\
\nummm $E \gets E \setminus \{(a,s)\}$.\\
\numm {\bf Done}\\
\num {\bf Done}
\end{niralgo}
}
\subsection{Attacking the Belief Propagation decoding algorithm} In the simple
scenario depicted in Figure~\ref{fig::simpleAttack}, the adversary can corrupt
each decoded symbol by altering a single encoding packet overall. In general, it
would be interesting to find the best possible strategy for the adversary given
that the adversary works either {\em offline} or {\em online} and is {\em
uniform} or {\em selective}.
\def\FIGA{
\begin{figure}
\centering
\includegraphics[scale=0.55]{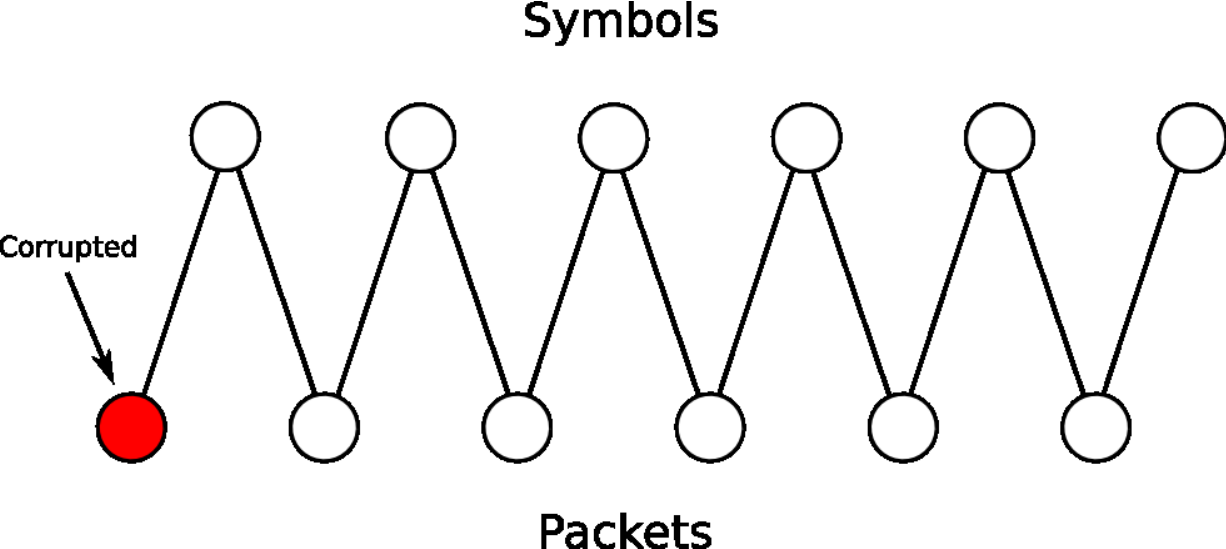}
\caption{Simple Attack scenario}\label{fig::simpleAttack}
\end{figure}
}
\noindent {$\bullet$ Offline vs. Online adversaries}. An {\em offline} adversary
knows, in advance, the graph generated at the {\em receiver}. 
An {\em online} adversary must base his decision to corrupt/inject
a single packet on the information available from the packets that traversed the
system so far.

\noindent {$\bullet$ Uniform vs. Selective adversaries}. A
$c$-bounded {\em uniform} adversary simulates random noise by uniformly
corrupting at most a fraction $c$ of all the received packets. Bare in mind that such an
adversary, though he cannot choose {\em which} packets to corrupt, can choose
{\em how} to corrupt them, negating simple solutions such as using {\sc crc} or
hashes. In contrast, a {\em selective} $c$-bounded adversary can choose,
non-uniformly, which packets to corrupt (and, of course, how to corrupt the
packets). It seems that the offline, selective adversarial model is the most
severe model, and we will target our results appropriately. We will present more
efficient solutions for a weaker model when applicable.

Under the definitions above, an interesting question is what would be the
optimal strategy for a given adversary in order to corrupt the largest number of
decoded symbols, given that the adversary is $c$-bounded. This immediately
translates to an upper bound on the number of symbols the adversary can corrupt,
a bound which may be employed in devising techniques to overcome the adversary.

We illustrate the vulnerability of the Belief Propagation decoder by presenting
the following attacks, using an online, selective adversary. Note that we believe that the specific attacks we list and prove are not the most severe; our tests show that corrupting a very small (constant) portion of the packets corrupts almost half of the symbols.

\noindent {\bf $\bullet$ The Vanishing Symbol Attack}. When using the Robust
Solition distribution to generate packets, one may calculate the fraction of the
packets in which each input symbol participates. Assuming that each symbol
participates in the generation of a fraction of $c$ packets, a simple online
{\em selective} $c$-bounded adversary can remove all traces of the symbol from
the system: fix an input symbol, $B$. The adversary will then remove from each
packet in which $B$ participated the indication that $B$ was {\sc xor}ed into
the packet. The decoder will then never successfully decode the entire message,
as $B$ will always be missing.

In \cite[Theorem 13]{LT}, it is shown that the average degree of a packet,
when using the Robust Soliton distribution, is in $O(\log(k/\delta))$ (where
$1-\delta$ is the probability of successful decoding). As the input symbols for
each packet is chosen uniformly, each input symbol has a probability of
(approximately) $\frac{\log(k/\delta)}{k}$ to be chosen for each packet. This
further implies that a $c$-bounded {\em online, selective} adversary may
prevent the receiver from successfully decoding (approximately)
$\frac{ck}{\log(k/\delta)}$ symbols. We note that this number of symbols is only a
gross estimate, as symbols with smaller degrees than the average are
numerous.  

\noindent {\bf $\bullet$ Odd packets attack}.
The following simple {\em online, selective} adversary can corrupt all decoded
symbols; Consider corrupting all packets which have an odd degree, i.e.,
connected to an odd number of symbols. Corrupt each packet by flipping all bits
(or a subset thereof). Using a simple inductive argument, we are able to prove
that the resulting decoded symbols from the Belief Propagation decoder will all
be flipped.
\begin{proofsketch}
The proof is by inspecting the sets of odd and even degree packets, throughout
the execution of the Belief Propagation decoder. Packets of odd degree are
corrupted, and packets of even degree are not. Moreover, as each packet moves
from one set to the other, all the packets in the odd degree set remain
corrupted and those in the even degree set remain correct. Since each symbol is
eventually decoded by copying a packet of degree one, all decoded symbols will
be corrupted.
\end{proofsketch}

The following proposition shows that indeed, when using the Robust Soliton distribution from~\cite{LT}, the expected fraction of such odd degree packets is less than one third, and the adversary can corrupt all symbols with a non-zero probability.
\begin{proposition}\label{prop. one-half}
When using the Robust Soliton distribution, the odds packets attack has a probability of at least one half to corrupt all decoded symbols.
\end{proposition}
\begin{proof}
We start by analyzing the Ideal Soliton distribution from~\cite{LT}, which
specifies that, for $k$ input symbols, the degree distribution of each
encoded packet is
$$
P[degree=i]=\rho(i)=\left\{   
\begin{array}{ll}
\frac{1}{k} & i=1\\
\frac{1}{i(i-1)} & i\ge 2
\end{array}
\right.
$$
The probability that for a given packet $p$, the degree is odd is
\begin{eqnarray*}
P[degree(p)\textrm{ is odd}] &=& \sum_{i\ge 1, i \textrm{ is odd}}^k\rho(i) \\
&\le&
\frac{1}{k} + \frac{1}{3\cdot2}+\frac{1}{5\cdot4} + \frac{1}{7\cdot6}+\cdots \\
&=& \frac{1}{k}+\sum_{i=2}^\infty \frac{(-1)^i}{i} = \frac{1}{k}+1-\ln(2).
\end{eqnarray*}
We get that for $k\ge38$, the probability for each packet to be of odd degree is
less than $\frac{1}{3}$. Now, consider a binomial random variable, $X\sim
B(n,p)$, such that $n$ equals the number of packets and $1/6 \leq p \leq 1/3$.
$X$ represents the number of packets of odd degree.
The adversary can successfully corrupt all symbols, as long as $X\leq n/3$.
Using the normal approximation for the binomial distribution (where $Z\sim
U(0,1)$), we get that:
$$ 
P[X\leq n/3] \approx P[Z\leq \dfrac{n/3-np}{\sqrt{np(1-p)}}] \ge \dfrac{1}{2}.
$$
Therefore, the adversary has a probability of at least half to corrupt all
decoded symbols by corrupting at most one third of the packets.
\end{proof}
\section{Corruption Resilient Fountain Codes} \label{s:local}

We start the presentation of our codes by discussing the
encoding phase. We then proceed to establish the necessary tools required for successfully decoding an encoded message, and present several decoding alternatives, discussing the merits of each.

\subsection{Encoding}
\label{sec:enc}
\def\FIGB{
\begin{figure}
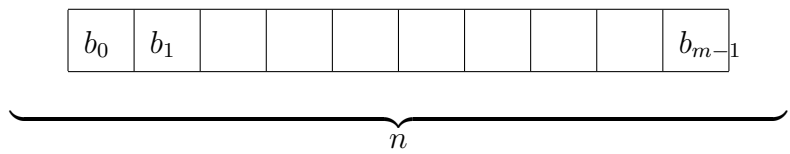

\centering
\begin{tabular}[]{|p{13pt}|p{13pt}|p{13pt}|p{13pt}|p{13pt}|p{13pt}|p{13pt}|p{13pt}|p{13pt}|p{13pt}|}
\hline
$b_0$ & $b_1$ & & & & & & & & $b_{m-1}$ \\
\hline
\end{tabular}
$\underbrace{\quad\quad\quad\quad\quad\quad\quad\quad\quad\quad\quad\quad\quad\quad\quad\quad\quad\quad\quad\quad\quad\quad\quad\quad\quad}_{\textstyle n}$
\caption{Splitting a message}
\label{fig:split}
\end{figure}
}
To encode a given message of $n$ bits, split the message into $m$ blocks (the
input symbols), $b_0, b_1, \ldots b_{m-1}$, each of length $k = \frac{n}{m}$
bits. For each $0\le i < m, 0 \le j < k$, let $b_i^j$ be the $j$'th bit of the $i$'th
message piece.

For a $k$ dimensional vector $\overline{v} = (v_0, v_1, \ldots, v_{k-1})$ over
$\mathbb{GF}(2)$ define the characteristic boolean function (or linear form)
$f_{\overline{v}}:\mathbb{GF}(2)^k\rightarrow \mathbb{GF}(2)$ in the following way:
$f_{\overline{v}}(x_0, x_1, \ldots, x_{k-1})= \bigoplus_j x_j v_j$ (in other
words, the inner product of $x$ and $v$ over $\mathbb{GF}(2)$).

To generate a packet $p$, randomly select a $k$ dimensional vector
$\overline{r}\in \mathbb{GF}(2)^k$ (the distribution used to sample
$\overline{r}$ will be discussed later), and set $p = \langle \overline{r},
f_{\overline{r}}(b_0), f_{\overline{r}}(b_1), \ldots, f_{\overline{r}}(b_{m-
1})\rangle$. Note that, for brevity, each packet is assumed to be of length
$k+m$ bits; later, we show how to reduce the amount of redundancy needed from
$k$ to $\log^2{k}$ bits, in essence by compressing the vector $\overline{r}$. We define the
following:

\begin{defini}
Two packets, $p_1$ and $p_2$, are termed {\em independent} if their associated
$\bar{r}$ vectors are independent over $\mathbb{GF}(2)$.
\end{defini}

Alternatively, a more efficient decoding (hardware-wise) is achieved by setting
$v_i = \langle b_0^i, b_1^i, \ldots, b_{m-1}^i \rangle$ and then 
$p = \langle \tilde{r}, f_{\tilde{r}}(v_0), f_{\tilde{r}}(v_1),
f_{\tilde{r}}(v_{k-1}) \rangle$, where $\tilde{r}\in \mathbb{GF}(2)^m$. This representation is more efficient as it
is faster to {\sc xor} entire {\em words} than individual {\em bits}.  We 
note that such encoding is similar to the one presented in \cite{LT}. 
The decoding procedures below are applicable for both encoding
alternatives. For brevity, we only discuss the first.
\subsection{Conditions for Linear Independence}\label{subsec. conditions}
Before introducing the decoding algorithms, we prove the following three Lemmas, which discuss the number of uncorrupted packets required in order to have a set $k$ independent equations. These lemmas are required in order to determine how many packets should a decoder collect before attempting to decode, depending on its knowledge on the number of corrupted packets and, of course, the type of the adversary (random or selective). The first considers the probability that a uniformly random linear system is not of full rank.  
\begin{lemma}\label{lemma:prob_uni}
The probability that a system $M$ of $m$ vectors of dimension $k$ ($m\ge k$),
chosen at random uniformly and independently over $\mathbb{GF}(2)$, is not of full rank $k$ is at
most $2^{k-m}$.
\end{lemma}
The importance of Lemma \ref{lemma:prob_uni} is clear. Suppose a decoder collects $k+\epsilon$ uncorrupted packets, $\epsilon \ll k$,  whose coefficients are drawn from the uniform distribution. Then, with probability at least $1-2^{-\epsilon}$, the packets include $k$ independent equations. 
\begin{proof} (Adapted from \cite{RAPTOR}): enumerate the vectors of $M$
arbitrarily. Consider the case in which the first $k$ vectors are not
independent. If so, these vectors span, at most, a half plane of dimension $k-
1$. For the entire system $M$ to be of rank less than $k$, each remaining vector
should be in this half plane. The probability of each of these $m-k$ vectors to
be uniformly chosen from this half place is at most $1/2$. Therefore, the
probability that all of the remaining $m-k$ vectors are chosen from the same
half plane is at most $2^{k-m}$.
\end{proof}
The next lemma, however, stands at the basis of our sparsity result. It states that even when the coding vectors are sparse (with approximately $\log k$ non-zeros for each vector), the same rank result holds. With this lemma, it will be easy to show that a more efficient decoding is possible (since sparse matrices are less complex to invert \cite{sparse_linear}) without harming the strength of the code.
\begin{lemma}\label{lemma:prob_log}
The probability that a system $M$ of $m$ vectors of dimension $k\leq m$
over $\mathbb{GF}(2)$,
where each coordinate is $1$ with probability $p$, independently of the others, with
$\frac{\log{k}+c}{k} \leq p \leq 1- \frac{\log{k}+c}{k}$, where
$c\rightarrow\infty$ slowly, is not of full rank
$k$ is at most $2^{k-m}$.
\end{lemma}
We term the distribution in Lemma \ref{lemma:prob_log} the {\em log distribution}.
\begin{proof}
By \cite[Theorem 1]{Cooper00}, the required probability, $P(k,m)$, as $k\rightarrow \infty$, is:
\[ 
\lim_{k \rightarrow \infty } P[M\textrm{ is not of full rank}] = 
1-\prod_{j=m-k+1}^\infty \left(1-2^{-j}\right)
\]

Let $x_i = (1-2^{-i})$, and set $B_{n,l} = \prod_{j=n}^{n+l} x_j$. We will first
show that $\forall n,l: B_{n,l} > x_{n-1}$. The proof is by induction over $l$; it
can easily be verified that for all $n$, $B_{n,0} = x_n > x_{n-1}$. Assume the
relation holds for all $n$ and for a given $l$. We will show the relation holds
for all $n$ and $l+1$: 
\begin{eqnarray*}
 B_{n,l+1} &=& x_n \underbrace{x_{n+1} x_{n+2} \cdots x_{n+l} x_{n+l+1}}_{B_{n+1,l}} 
 = x_n \cdot B_{n+1,l} \\
& \geq & x_n\cdot x_n  \;\;\;\textrm{(using the inductive assumption)}\\
 &=& \left(1-2^{-n}\right)^2  = \frac{2^{2n} -2^{n+1} + 1 }{2^{2n}} \\
 &=& (1-2^{-n+1}) + 2^{-2n} 
 =  x_{n-1} + 2^{-2n} > x_{n-1}
\end{eqnarray*}

Now, for a given $n$, as the series $\{B_{n,l}\}_{l=0}^{\infty}$ is
monotonically decreasing, and according to the above is bounded by $x_{n-1}$,
there exists a constant $L(n)$, such that $\lim_{l\rightarrow \infty}B_{n,l} =
L(n)$. We will now show that $L(n) \ge x_{n-1}$. Assume, towards contradiction,
that $L(n) < x_{n-1}$. From the definition of the limit, for each $\epsilon$
there exists an $l$ such that $|B_{n,l}-L(n)| < \epsilon$. Take $\epsilon =
x_{n-1}-L(n)$, and we get that there exists an $l$ such that:
$|B_{n,l}-L(n)| < x_{n-1}-L(n)$. Removing the absolute value, as the
series converges from above, results in $B_{n,l}-L(n) < x_{n-1}-L(n)$, that is,
 $B_{n,l} < x_{n-1}$, which is a contradiction. 
As a result, we get that 
\begin{eqnarray*}
 P(k,m) &=&1-\prod_{j=m-k+1}^\infty \left(1-2^{-j}\right)
= 1-\lim_{l\rightarrow \infty} B_{m-k+1,l} 
\\ &<& 1- x_{m-k} 
= 2^{k-m}
\end{eqnarray*}
\end{proof}
The third lemma considers the possibility that \emph{any subset of sufficient size}, out of the $|N|$ collected packets, includes $k$ independent equations. It will be applicable when one considers decoding under selective attacks, as in these cases it is no longer true that the uncorrupted packets arriving at the decoder have the originally intended distribution, and the previous two lemmas do not apply directly. Let $h(p)=-p\log p -(1-p)\log(1-p)$ denote the binary entropy function. We have the following.
\begin{lemma}\label{lemma:prob_selective}
Let $S = \{p_j | p_j=(r_j, f_{\overline{b_l}}(r_j))\}$ be a set of packets,
all generated using the uniform distribution, or the log distribution of
Lemma~\ref{lemma:prob_log}. Assume $|S|=a\cdot k$, $a \in \R^+$. Then, with probability at least $1-2^{-k(a-b-1-ah(b/a))}$, where $b\in \R^+$, $b<a$, every subset of $(a-b)k$ packets out of $S$ contains an independent subset of size $k$.
\end{lemma}
For example, for $b=1$ and $a=7$, this probability is approximately $1-2^{-0.123k}$. Since $bk$ will later denote the number of corrupted packets, and $a$ will be a constant chosen by the decoder, depending on $b$, we write $a(b)$ to denote the value of $a$ chosen by the decoder to ensure that $a-b-1-ah(b/a) > 0$.
\begin{proof}
Using the bounds given in Lemmas~\ref{lemma:prob_uni} and~\ref{lemma:prob_log},
 we have:
\begin{eqnarray*}
Pr\big[\exists S'\subset S : |S'| = |S|-bk \land 
rank(S')<k \big] 
&\leq & \binom{|S|}{|S|-bk}2^{k-(|S|-bk)} \\
&=& \binom{ak}{bk}2^{k-ak+bk}  
\leq 2^{ak\cdot h(b/a)}2^{k-ak+bk} \\
&=&2^{-k(a-b-1-a h(b/a))}.
\end{eqnarray*}
\end{proof}
\vspace{-1cm}
\subsection{Decoding}\label{subsec. decoding}
We present several possible ways to decode a value, where the trade off between
the number of packets which need to be collected and the decoding time is
investigated. 
We will limit the discussion to decoding a given message block, $b_l$, where all
message blocks may be decoded in parallel, using the same technique. Throughout the discussion, it is beneficial to consider both the total number of collected packets at the decoder (both original and corrupted) and the number of corrupted packets collected. Hence, we let $N$ be the set of all collected packets, and $f$ denote the number of corrupted packets collected. Note that $c=\frac{f}{|N|}$.

The first decoding algorithm is a majority test, included here to illustrate a simple, yet deterministic and polynomial, decoding procedure. While the complexity is polynomial in $k$, it is not rate-optimal in terms of the number of packets that needs to be collected for successful decoding. The randomized algorithm suggested later is superior both in expected rate and expected complexity.
 
\noindent {\bf Majority voting}. The decoding algorithm is depicted in
Figure~\ref{figure:majority}, and is applicable to both the uniform and
selective adversarial models. To reconstruct a given message block, $b_l$, given
that at most $f$ faults occurred, we need to collect $2f+1$ pairwise disjoint
sets of packets, $S_1, S_2, \ldots, S_{2f+1}$, such that each set contains
exactly $k$ independent packets. In Lemma \ref{lemma:prob_uni} it was shown that indeed an addition of a constant number of packets to a set of $k$ packets, ensures (with high probability) $k$ independent equations. This applies to a random adversary, as the packets it corrupts are randomly selected, \emph{hence the distribution of the uncorrupted packets seen by the decoder is still random}. However, this is not true for a selective adversary. Such an adversary has the ability to chose which packets to corrupt, and, consequently, change the distribution of the uncorrupted packets seen by the decoder, such that Lemmas \ref{lemma:prob_uni} and \ref{lemma:prob_log} will not hold as is. In this case, Lemma \ref{lemma:a_selective} comes in handy, as it assures that if enough packets are collected, with high probability \emph{any subset of $(a-b)k$} packets will give the desired result. Clearly, more packets need to be collected ($(a(b)+b) k$ packets, where $f=b\cdot k$ is the number of corrupted packets), but the trade-off allows handling a selective adversary.

From each set, $S_j$, we can reconstruct an $s_j$ as a candidate message block.
It then follows that the majority of the values is the correct message block. In
other words, $b_l = \argmax_{s_j}|\{s_i : s_i = s_j\}|$.
\def\FIGC{
\begin{figure}
\begin{displaymath}
\left.
\begin{array}{rl}
1. & \kkk{1}\\
2. & \kkk{2}\\
\vdots\\
2f+1. & \kkk{2f+1}
\end{array}
\right\} majority
\end{displaymath}
\caption{Using majority logic to decode a message piece}
\label{figure:majority}
\end{figure}
}
Note that to ensure that $|N|$, the number of packets collected,
will suffice to compose $(2f+1)$ sets of $k$ packets, we need $c \leq \frac{1}{2k}-\frac{1}{2N} \approx \frac{1}{2k}$. Thus, Majority voting decoding may be used only when the adversary is at most $\frac{1}{2k}$-bounded.
The running time of such an algorithm is dominated by the need to solve the $(2|N|c+1)$ equation sets, which takes, in general, $O(2Nck^3)$. However, using the \emph{sparse coding vectors} of Lemma \ref{lemma:prob_log}, the complexity is $O(2Nck^2 \log k \log\log k)$, as sparse equations can be solved more efficiently \cite{sparse_linear}. 

At the other extreme of the trade-off between complexity and rate, we present an asymptotically rate optimal algorithm using exhaustive search.

\noindent {\bf Exhaustive search algorithm}. 
Similar to the Majority test algorithm, we assume the number of packets corrupted, $f$, is known to the decoder, and his goals are to both \emph{decide how many packets to collect and, of course, decode the original block}.
\begin{lemma}\label{lemma:a_random}
Let $N = \{p_j | p_j=(r_j, f_{\overline{b_l}}(r_j))\}$ be a set of packets, such
that $|N|\ge k+2f+\epsilon$, drawn using either the uniform distribution or the
log distribution. Define the following matrix, $\hat{A} = (r_j)$ and
let $\bar{b} = (f_{\overline{b_l}}(r_j))$. Assuming that no more than $f<k$
packets are corrupted by a uniform adversary, and that with very high
probability $k+\epsilon$
packets contain a subset of $k$ independent packets. Then $b_l$ is the
only solution to the following equation system which satisfies at least
$k+f+\epsilon$ equations:
$\hat{A}\cdot\bar{x} = \bar{b}$.
\end{lemma}
\begin{proof}
Knowing $f$, $k$ an $\epsilon$, the decoder collects $k+2f+\epsilon$ packets. Since there are at most $f$ packets corrupted, there is at least one subset of $k+f+\epsilon$ un-corrupted packets. Denote it by $S$. There is at least one solution (the true one, denoted $\bar{x}$) satisfying all the equations in $S$. However, since $S$ includes at least $k+\epsilon$ uncorrupted packets, which, in turn, with very high probability, contain $k$ independent equations, $\bar{x}$ is actually the only solution satisfying all equations in $S$.

It remains to show that the decoder cannot find a different set of $k+f+\epsilon$ packets, all satisfying a different solution. Consider a set of $k+f+\epsilon$ packets, for which there is a solution $\tilde{x}$ satisfying all $k+f+\epsilon$ equations. If this set includes some subset of $k+\epsilon$ un-corrpted packets, then, since this subset includes $k$ independent un-corrupted equations, $\tilde{x}=\bar{x}$. If there is no such subset, then there are more than $f$ corrupted equations, which is a contradiction.   
\end{proof}
Note that $f<k$ and hence $c=\frac{f}{k+2f+\epsilon}<\frac{1}{3}$. Furthermore, 
note that the proof relies on the uniformity of the adversary since it assumes any set of uncorrupted packets of size $k+\epsilon$ satisfies Lemma \ref{lemma:prob_uni} or \ref{lemma:prob_log}, hence include, with hight probability, $k$ independent equations. When the adversary is selective, this is not necessarily the case, as the adversary can choose which packets to corrupt, and thus inflect a different distribution of uncorrupted packets at the decoder. In that case, the following lemma will hold.
\begin{lemma}\label{lemma:a_selective}
Assume an encoder generates packets according to the uniform or the log distribution, and that no more than $bk, b\in \R^+$, packets are corrupted by a selective offline adversary. Let $N = \{p_j | p_j=(r_j, f_{\overline{b_l}}(r_j))\}$ be the set of packets collected at the decoder, with $|N|\ge (a+b)k, a>b$. Define $\hat{A} = (r_j)$ and
let $\bar{b} = (f_{\overline{b_l}}(r_j))$. Then, with probability at least $1-2^{-k(a-b-1-ah(b/a))}$, $b_l$ is the
only solution to the following equation system which satisfies at least
$ak$ equations:
$\hat{A}\cdot\bar{x} = \bar{b}$.
\end{lemma}
\begin{proof}
The proof is similar to that of Lemma \ref{lemma:a_random}. The decoder collects $(a+b)k$ packets. Since there are at most $bk$ corrupted ones, the remaining $ak$ include an uncorrupted subset of size $(a-b)k$ and by Lemma \ref{lemma:prob_selective} include $k$ independent packets. This constitutes a valid solution. Any other solution which satisfies at least $ak$ equations must coincide with that solution since at least $(a-b)k$ of the $ak$ equations are from uncorrupted packets and hence include $k$ independent equations. 
\end{proof}
It follows from Lemma~\ref{lemma:a_random} that the solution, $\bar{x}$, to the system
of equations $\hat{A}\cdot\bar{x} = \bar{b}$, which correctly solves the largest
number of equations is the (only) correct solution. Hence, an algorithm which
solves the following optimization problem would have been invaluable: given a
system of equations over $\mathbb{GF}(2)$, find the best possible $\bar{x}$,
which solves the maximal number of equations. Obviously, this is an NP-complete
problem (by a simple reduction from the max-cut problem,
see~\cite{AK}).

Given Lemma~\ref{lemma:a_random} and the above argument, a simple exhaustive search
over all possible $b_l$ values yields the correct answer, which satisfies at
least $k+f+\epsilon$ equations
out of $N$ (or $(a-b)k$ out of the $(a+b)k$ in the selective model, according to Lemma \ref{lemma:a_selective}). Exhaustive search decoding is applicable to all $c$-bounded
adversaries. Nevertheless, as such a search is exponential in $k$ -- in fact, the running time of such an algorithm is in $O(2^k+k\cdot|N|)$ -- it may not be applicable in all situations. 

Next, we present better decoding algorithms, which trades decoding time for
increased amounts of packets, assuming that the adversary is bounded by 
values smaller than $1/3$.
\def\ALGB{
\begin{niralgo}{Randomized decoding algorithm}{algo:rand:decode}{8cm}{!htb}
xxx \= xxx \= xxx \= xxx \= xxx \= \kill
$N:$ A set of $g(k)+k+f+\epsilon$ packets\\
\\
{\bf good}$(s):=$ $s$ satisfies at least $k+f+\epsilon$ \\
\hspace{1.75cm} equations out of $N$\\
\\
{\bf Solve}$(S)$:\\
\num Try to solve the equation system induced by $S$.\\
\num Return the solution, or $\bot$ if no solution exists.\\
\\
{\bf Repeat}\\
\num $S \gets$ a random set of $N$, $|S|=k+e$\\
\num $s \gets$ {\bf Solve}$(S)$\\
\num {\bf if} $s\neq \bot \land good(s)$ {\bf then}\\
\numm {\bf return} $s$\\
\num {\bf fi}\\
{\bf Done}
\end{niralgo}
}
\noindent 
{\bf Randomized decoding algorithm}.
We use randomization in order to reduce the decoding complexity, given that the
adversary is $c$-bounded (for an appropriate $c$, to be defined later), and may
only corrupt at most $f$ packets. The algorithm is depicted in
Figure~\ref{algo:rand:decode}, and suits the uniform adversarial model, as we follow Lemma~\ref{lemma:a_random} from the previous section. An extension to the selective adversary is trivial, with, of course, the right choice of parameters, as given in Lemma \ref{lemma:a_selective}. 

Assume that we have collected $N$
packets, such that $|N| =
g(k)+k+f+\epsilon \ge k+2f+\epsilon$ for some $g(x)\ge f$ to be defined later.
Each subset of $k+\epsilon$ uncorrupted packets has a very high probability of containing a subset of $k$ independent packets. Let this probability be $p_\epsilon$, which can be calculated according to Lemma~\ref{lemma:prob_uni}. The algorithm will work as follows: choose a random subset, $S \subset N$, such that $|S| = k+\epsilon$. Let $s$ be the unique solution (if such exists) to the equation $\hat{A}\cdot \bar{x} = \bar{b}$, defined by $S$. If the obtained solution, $s$, satisfies more than $k+f+\epsilon$ equations out of $N$, then, according to Lemma~\ref{lemma:a_random}, $b_l = s$ and we are done.

Now, let $p_k$ be the probability of choosing a subset of $N$ with no corrupted
packets (hence, the probability of finding the right solution). It then follows that the expected number of iterations of the algorithm is $\frac{1}{p_k p_\epsilon}$.
As each iteration involves solving an equation system of size $k+\epsilon$
and validating the solution against at most $|N|$ packets, the running time of
each iteration is $O(k^2\log k \log \log k+|N|\log k)$, remembering there are approximately $\log k$ non-zeros coefficients in each equation. This yields an overall
expected running time of $O(\frac{k^2\log k \log \log k+|N|\log k}{p_k\cdot
p_\epsilon})$.

By Lemma \ref{lemma:prob_uni}, $p_\epsilon>1-2^{-\epsilon}$. Let us now approximate $p_k$.
\begin{eqnarray*}
p_k & =& \frac{ \binom{g(k)+k+\epsilon}{k+\epsilon}}{\binom{g(k)+k+f+\epsilon}{k+\epsilon}}
= \frac{\frac{(g(k)+k+\epsilon)!}{g(k)!}}{\frac{(g(k)+k+f+\epsilon)!}{(g(k)+f)!}} \\
& = & \frac{(g(k)+1)\cdots(g(k)+f)}
{(g(k)+k+1+\epsilon)\cdots(g(k)+k+f+\epsilon)}\\
& \ge & \left(\frac{g(k)}{g(k)+k+\epsilon}\right)^f =
\left(1+\frac{k+\epsilon}{g(k)}\right)^{-f}  >
\exp\left(-f\dfrac{k+\epsilon}{g(k)}\right).
\end{eqnarray*}
We can thus choose $g(k)$ according to our needs --
either minimizing decoding time or minimizing the number of packets we need to
collect. For example, choosing $g(x) > \frac{f\cdot(k+\epsilon)}{\log{b}}$, for
a constant $b$, results in $p_k>1/b$. Such a choice minimizes the run time to $O(k^2 \log k \log \log k)$, at
the expense of having to collect many messages ($g(x) \approx kf$). Furthermore,
such a decoding algorithm is only relevant when the adversary is at most
$\frac{1}{k}$-bounded.

\noindent {\bf Reducing packets size}. In Section~\ref{sec:enc}, we assume that
each packet generated is of size $k+m$ bits, where $m$ is the size of each
message block and $k$ is the number of blocks. The $k$ bits of redundancy are
used to denote which blocks participated in the creation of the packet. In
practice, these $k$ bits may be compressed, using several techniques, into a
logarithmic size. First, when choosing which blocks to {\sc xor} in a uniform
fashion, the sender can use a {\sc prng} to generate the required distribution,
as proposed in~\cite{LT}. The sender will use a different random seed for the
{\sc prng} for every packet generated, and only attach the seed to the packet.
A seed of size logarithmic in $k$ suffices, and the receiver can
proceed to recover which blocks participated in creating a specific packet using
the seed embedded in the packet.

A different approach, in which the use of a {\sc prng} (pseudo random number
generator) is not required, can be achieved by using 
the log distribution to 
select blocks for each packet;
first generate a binary vector, $r$,
of length $k$, in which each index is $1$ with probability
$\frac{(1+\delta)\log{k}}{k}$, for some small constant $\delta$. When sending
each packet, do not attach $r$ to the packet. Rather, attach the indices of
$r$ which contain $1$. 
\section{Concluding Remarks}
\label{s:cr}
Error correcting codes, erasure correcting codes, communication networks and
distributed computing are tightly coupled. The replication techniques used for
obtaining fault-tolerance in distributed computing may be replaced by error and
erasure correcting techniques. Beyond the memory overhead benefit, the dispersal
of information can be useful to protect and hide clear-text values when necessary. 
In this work, we have presented efficient encoding and decoding schemes that enhance the
correctability of well known rateless erasure correcting codes.\\
\noindent
{\bf Acknowledgments}. It is a pleasure to thank Michael Mitzenmacher for 
helpful inputs and for pointing out relevant related works, and Shachar Golan
for helpful comments.

\newpage
\ALGA
\ALGB
\FIGA
\FIGB
\FIGC
\end{document}